\documentclass{conm-p-l}

\copyrightinfo{2010}{}

\usepackage{graphicx}

\usepackage{amssymb,amsmath}

\newtheorem{theorem}{Theorem}[section]

\theoremstyle{definition}

\theoremstyle{remark}
\newtheorem{remark}[theorem]{Remark}

\numberwithin{equation}{section}

\newcommand{\e}{\epsilon}

\renewcommand{\k}{\kappa}

\newcommand{\ra}{\rightarrow}
\newcommand{\al}{\alpha}
\newcommand{\be}{\beta}
\newcommand{\sg}{\sigma}

\newcommand{\pa}{\partial}

\newcommand{\bv}{\bar{v}}

\newcommand{\bp}{\bar{p}}

\newcommand{\non}{\nonumber}

\begin{document}

\title[Type II 3D Shears]
{Stability Criteria and Turbulence Paradox Problem For Type II 3D Shears}

\author{Y. Charles Li}

\address{Department of Mathematics, University of Missouri, 
Columbia, MO 65211, USA}

\curraddr{}
\email{liyan@missouri.edu}

\thanks{}

\subjclass{Primary 76, 37; Secondary 35}
\date{}

\dedicatory{}

\keywords{Type II 3D shear, turbulence paradox, channel flow, hydrodynamic 
stability, phase space.}

\begin{abstract}
There are two types of 3D shears in channel flows: ($U(y,z),0,0$) and 
($U(y),0,W(y)$). Both are important in organizing the phase space structures 
of the channel flows. Stability criteria of the type I 3D shears were studied in [Li, 2010]. 
Here we study the stability criteria of the type II 3D shears. We also provide more support 
to the idea of resolution of a turbulence paradox, introduced in [Li and Lin, 2010], by 
studying a sequence of type II 3D shears.
\end{abstract}

\maketitle

\section{Introduction}

Studying the phase space structures of channel flows is an important and emerging area. Most of the 
existing works in this area is numerical. Like every other dynamical system study in a phase space, 
non-wandering objects like fixed points, periodic orbits etc. play a fundamental role in organizing the 
phase space structure. When the Reynolds number is infinite (i.e. zero viscosity), the corresponding 
3D Euler equations have two types of steady shears as fixed points in the phase space. 
\begin{itemize}
\item Type I 3D shears ($U(y,z),0,0$),
\item Type II 3D shears ($U(y),0,W(y)$);
\end{itemize}
where the boundaries of the channel are in the $y$-direction. When the Reynolds number is not infinite
but large, these shears turn into slowing drifting states under the corresponding 3D Navier-Stokes 
dynamics. By the concepts of rate condition and normal hyperbolicity of Fenichel \cite{Fen74} 
\cite{Fen77}, these slowing drifting states can be crucial in the transition to turbulence. Also the 
3D Navier-Stokes dynamics has fixed points like the so-called lower and upper branches \cite{Nag90}. 
When the Reynolds number approaches infinity, the lower branch approaches one of the type I 3D shears. 
How to distinguish this particular one from the rest of the type I 3D shears is an interesting question 
(also posted in a list of problems by Yudovich \cite{Yud03}). A condition satisfied by this particular
type I 3D shear was derived in \cite{LV09}. Again like every other dynamical system study in a phase 
space, the stability of these 3D shears is crucial in understanding the phase space structure. The 
stability criteria for type I 3D shears were studied in \cite{Li10}. Here we shall study the 
stability criteria for type II 3D shears. It turns out that the linearized 3D Euler equations 
can be casted into formally the same form as the classical Rayleigh equation for 2D shears ($U(y),0$). 
But the nature of the stability criteria is fundamentally different from that for the 2D shears. One 
can ask the question: What is the ``percentage'' among e.g. all type I 3D shears, that is unstable? 
The author's conjecture is:
\begin{itemize} 
\item The unstable percentage of type II 3D shears $>$ the unstable percentage of type I 3D shears
$>$ the unstable percentage of 2D shears.
\end{itemize}
Finally, we shall provide more support to the idea of resolution of the turbulence paradox, introduced 
in \cite{LL10} by studying a sequence of type II 3D shears. The turbulence paradox is also called 
Sommerfeld paradox which roughly says that the linear shear in the plane Couette flow is linearly stable 
for all values of the Reynolds number, while in experiments, transition from the linear shear to 
turbulence occurs when the Reynolds number is large enough. For more details on the turbulence paradox
and our resolution, see Section \ref{para} and the paper \cite{LL10}.

\section{Necessary Conditions For Instability}

The inviscid channel flow is governed by the 3D Euler equations
\begin{equation}
\pa_t u_i + u_j u_{i,j} = - p_{,i}, \quad u_{i,i} = 0 ; 
\label{Euler}
\end{equation}
where ($u_1,u_2,u_3$) are the three components of the fluid velocity along 
($x,y,z$) directions, and $p$ is the pressure. 
The boundary condition is the so-called non-penetrating condition
\begin{equation}
u_2(x, a, z) = 0, \quad u_2(x, b, z) = 0;
\label{bc}
\end{equation}
where $a<b$ are the boundary locations of the channel in $y$-direction.

We start with the type II 3D steady shear solutions of the 3D Euler equations:
\[
u_1 = U(y),\quad u_2 = 0, \quad u_3 = W(y), \quad p = p_0 \ (\text{a constant}).
\]
Of particular importance are those profiles which also satisfy the non-slip boundary condition
\[
U(a) = \al , \quad U(b) = \be , \quad W(a)=W(b)=0 ;
\]
where $\al$ and $\be$ are the velocities of the two walls of the channel. Such profiles may be the viscous 
limiting profiles when the viscosity approaches zero. Linearize the 3D Euler equations with the notations
\begin{eqnarray*}
& & u_1 = U(y) + \left [ e^{ik_1x +ik_3z -i\sg t} u(y) + c.c. \right ], \ u_2 = e^{ik_1x +ik_3z -i\sg t} v(y) + c.c.,  \\
& & u_3 = W(y) + \left [ e^{ik_1x +ik_3z -i\sg t} w(y) + c.c. \right ], \ p \ra p_0+\left [ e^{ik_1x +ik_3z -i\sg t} p(y) + c.c.  \right ];
\end{eqnarray*}
where $k_1$ and $k_3$ are the wave numbers, and $\sg$ is a complex constant; we obtain the linearized 3D Euler equations
\begin{eqnarray}
& & i(k_1U+k_3W-\sg ) u+ U'v  = - ik_1 p, \label{LC1} \\
& & i(k_1U+k_3W-\sg ) v   = - p' , \label{LC2} \\
& & i(k_1U+k_3W-\sg ) w +W'v  = - ik_3 p, \label{LC3} \\   
& & i k_1 u + v' + ik_3w = 0 . \label{LC4}
\end{eqnarray}
Two forms of simplified systems can be derived:
\begin{equation}
 v'' - \frac{k_1 U''+k_3W''}{k_1U+k_3W-\sg }v = (k_1^2+k_3^2)v, \label{ray1} 
\end{equation}
with the boundary condition $v(a)=v(b)=0$; and 
\begin{equation}
(k_1U+k_3W-\sg)^2 \left [ (k_1U+k_3W-\sg)^{-2} p' \right ]' = (k_1^2+k_3^2) p , \label{ray2}
\end{equation}
with the boundary condition $p'(a)=p'(b)=0$.

Notice that formally equation (\ref{ray1}) is in the same form as the clasical Rayleigh equation
where the 2D shear $U(y)$ is now replaced by $k_1U+k_3W$. We have the following extension of the 
classical Rayleigh's inflection-point theorem.
\begin{theorem}
If ($k_1,k_3$) is an unstable mode, then
\begin{equation}
\left ( k_1 U''+k_3W'' \right ) (y_*) = 0 \quad \text{for some } y_* \in (a,b). \label{gcon}
\end{equation}
\label{Raye}
\end{theorem}
\begin{remark}
Even though the Rayleigh's inflection-point theorem is a special case of the above theorem, the claim in 
the above theorem in general is fundamentally different from that of Rayleigh. Rayleigh's claim basically 
says that if ($U(y),0$) is linearly unstable, then $U(y)$ has an inflection point. Here on the other hand,
condition (\ref{gcon}) can be satisfied by most of type II 3D shears ($U(y),0,W(y)$). In fact, one can choose 
$k_1 = W''(y_*)$, $k_3 = -U''(y_*)$ for any $y_* \in (a,b)$, then (\ref{gcon}) is satisfied. Therefore, 
in general (\ref{gcon}) is a very weak necessary condition for the linear instability of ($U(y),0,W(y)$). 
This may indicate that type II 3D shears are more often to be unstable than 2D shears. The author's 
conjecture is that type I 3D shears are inbetween in terms of frequency of instability. 
\end{remark}
\begin{proof}
Multiplying (\ref{ray1}) by $\bv$,
\begin{equation}
\int_a^b \left [ |v'|^2 + (k_1^2+k_3^2) |v|^2 \right ] dy + \int_a^b \frac{k_1 U''+k_3W''}{k_1U+k_3W-\sg }
|v|^2 dy = 0,
\label{fr1}
\end{equation}
the imaginary part of which is 
\begin{equation}
\sg_i \int_a^b \frac{k_1 U''+k_3W''}{|k_1U+k_3W-\sg |^2 }|v|^2 dy = 0,
\label{fr2}
\end{equation}
where $\sg = \sg_r +i\sg_i$. Thus
\[
\left ( k_1 U''+k_3W'' \right ) (y_*) = 0 \quad \text{for some } y_* \in (a,b).
\]
\end{proof}

We also have the following extension of the Fj$\phi$rtoft's theorem. 
\begin{theorem}
If ($k_1,k_3$) is an unstable mode, then
\[
G''(y_0)\left [ G(y_0) - G(y_*) \right ] < 0 \quad \text{for some } y_0 \in (a,b),
\]
where $G = k_1U+k_3W$ and $y_*$ is the point at which $G''=0$ given by Theorem \ref{Raye}.
\end{theorem}
\begin{proof}
The real part of equation (\ref{fr1}) is 
\[
\int_a^b \frac{G'' (G-\sg_r)}{|G-\sg |^2}|v|^2 dy = - \int_a^b \left [ |v'|^2 + (k_1^2+k_3^2) |v|^2 \right ] dy .
\]
From (\ref{fr2}), one has 
\[
[\sg_r - G(y_*)] \int_a^b \frac{G''}{|G-\sg |^2}|v|^2 dy = 0 ,
\]
where $y_*$ is chosen to be the point given by Theorem \ref{Raye}, $G''(y_*) = 0$. The above two equations 
imply that 
\[
\int_a^b \frac{G''(y) [G(y)-G(y_*)]}{|G(y)-\sg |^2}|v|^2 dy < 0 , 
\]
and the theorem is proved.
\end{proof}

Next we show the extension of Howard's semi-circle theorem. 
\begin{theorem} 
If ($k_1,k_3$) is an unstable mode, then its corresponding unstable eigenvalue 
lies inside the semi-circle in the complex plane:
\[
\left ( \sg_r -\frac{M+m}{2} \right )^2 +\sg_i^2 \leq   \left (  \frac{M-m}{2}  \right )^2 , 
\]
where again $\sg = \sg_r +i\sg_i$, $M = \max_{y\in [a,b]} (k_1U+k_3W)$, and $m = \min_{y\in [a,b]} (k_1U+k_3W)$.
\end{theorem} 
\begin{proof}
Multiply (\ref{ray2}) with $\bp (G-\sg )^{-2}$, integrate by parts, and split into real and imaginary 
parts; we obtain that 
\begin{eqnarray}
& & \int_a^b G Q dy = \sg_r \int_a^b Q dy , \label{hw1} \\
& & \int_a^b G^2 Q dy  = 2\sg_r \int_a^b G Q dy + (\sg_i^2 - \sg_r^2) \int_a^b Q dy \non \\
& & = (\sg_r^2 +\sg_i^2) \int_a^b Q dy \label{hw2} 
\end{eqnarray}
by (\ref{hw1}), where
\begin{eqnarray*}
G &=& k_1U+k_3W , \\
Q &=& |G-\sg |^{-4} \left [ |p'|^2 + (k_1^2+k_3^2) |p|^2 \right ] .
\end{eqnarray*}
Let 
\[
M = \max_{y\in [a,b]} G , \quad m = \min_{y\in [a,b]} G ,
\]
then
\[
\int_a^b (G-m) (M-G) Q dy \geq 0 .
\]
Expand this inequality and utilize (\ref{hw1})-(\ref{hw2}), we arrive at the semi-circle 
inequality in the theorem.
\end{proof}

\section{Sufficient Conditions For Instability}

The zeroth mode ($k_1,k_3$)=($0,0$) is trivially neutrally stable, so our interest is focused 
upon non-zero modes. Without loss of generality, we assume $k_1 \neq 0$. Equation (\ref{ray1}) 
can be re-written in the form 
\[
v'' - \frac{U'' + \frac{k_3}{k_1}W''}{U + \frac{k_3}{k_1}W - \frac{\sg }{k_1}} v = 
k_1^2 \left [ 1+ \left (\frac{k_3}{k_1} \right )^2 \right ] v .
\]
We will keep $\k = k_3/k_1$ fixed and vary $k_1$. Denote by
\[
H = U + \k W, \ c =\frac{\sg }{k_1}, \ \al = k_1 \sqrt{1+\k^2} , 
\]
we obtain the following equivalent form of (\ref{ray1}),
\begin{equation}
 v'' - \frac{H''}{H-c }v = \al^2 v, \label{rie} 
\end{equation}
with the boundary condition $v(a)=v(b)=0$, which is in the same form as the classical Rayleigh equation 
\cite{LL10}. Thus we have the extension of the Tollmien's theorem on a sufficient condition for instability.
\begin{theorem}
If $H'(y) \neq 0$ for all $y \in (a,b)$, $H''(y_*) = 0$ for some $y_* \in (a,b)$, and the Sturm-Liouville 
operator 
\[
Lv = -v'' + \frac{H''}{H-H(y_*)}v
\]
has a negative eigenvalue under the Dirichlet boundary condition $v(a)=v(b)=0$; then the equation (\ref{rie})
has an unstable eigenvalue (in fact, an unstable curve $c=c(\al )$ for $\al$ in some interval).
\end{theorem} 

\section{Turbulence Paradox \label{para}}

Turbulence paradox is also called Sommerfeld paradox. It originated from Sommerfeld's analysis which concluded 
that the linear shear in plane Couette flow is linearly stable for all values of the Reynolds number; whereas 
in fluid experiments, perturbations of the linear shear often lead to transition to turbulence. Such a paradox 
is universal among fluid flows, e.g. pipe Poiseuille flow, plane Poiseuille flow etc.. A resolution of this 
paradox is given in \cite{LL10}. The main idea of the resolution is to show that even though the linear shear 
is linearly stable, states arbitrarily close to the linear shear can still be linearly unstable. Here different 
norms are crucial. It is shown in \cite{LL10} that the sequence of 2D shears ($U_n(y),0$) is linearly unstable 
for all $n$ and large enough Reynolds number (including infinity), where 
\begin{equation}
U_n(y) = y + \frac{A}{n} \sin (4n\pi y), \ \left ( \frac{1}{2}\frac{1}{4\pi} < A < \frac{1}{4\pi} \right ) .
\label{2s}
\end{equation}
Here $U_n(y)$ approaches the linear shear $U=y$ in the $L^\infty$ norm of velocity, but not in the 
$L^\infty$ norm of vorticity. Notice also that when the Reynolds number is not infinite but large, the 
shears ($U_n(y),0$) are not steady states rather slowly drifting states. The linear instability mentioned
above is predicted by the Orr-Sommerfeld operator (linearized 2D Navier-Stokes operator) when the 
shears ($U_n(y),0$) are viewed frozen. Such an instability is in the spirit of Fenichel's rate condition 
and normal hyperbolicity \cite{Fen74} \cite{Fen77}.  Such an instability is  also observed numerically 
\cite{LL10b}. 

In this paper, we would like to add more support to the idea of resolution mentioned above by considering 
the type II 3D shears ($U(y),0,W(y)$).

The viscous channel flow is governed by the Navier-Stokes equations
\begin{equation}
\pa_t u_i + u_j u_{i,j} = - p_{,i} +\e u_{i,jj} , \quad u_{i,i} = 0 ; 
\label{NS-Couette}
\end{equation}
where again ($u_1,u_2,u_3$) are the three components of the fluid velocity along 
($x,y,z$) directions, $p$ is 
the pressure, and $\e = 1/R$ is the inverse of the Reynolds number $R$. 
The boundary condition is
\begin{equation}
u_1(x, a, z) = \al , \quad u_1(x, b, z) = \be , \quad
u_j(x, a, z) = u_j(x, b, z) = 0, (j=2,3);
\label{BC-Couette}
\end{equation}
where $a<b$, $\al < \be$. For the viscous channel flow,
the type II 3D shears mentioned above are no longer fixed points, instead they 
drift slowly in time (sometimes called quasi-steady solutions):
\[
\left ( e^{\e t \pa_y^2 } U(y), 0, e^{\e t \pa_y^2 } W(y) \right ).
\]
By ignoring the slow drift and pretending they are still fixed points
(or by using artificial body forces to stop the drifting), 
their unstable eigenvalues will lead to transient nonlinear 
growths as shown numerically \cite{LL10b}. A better explanation here is to 
use the theory of geometric singular perturbation \cite{Fen74} \cite{Fen77} \cite{Fen79} \cite{LW97}.
The slowly drifting 3D shears altogether form a locally invariant 
slow (center) manifold. The normal direction growth rate (or decay rate)
of this slow manifold has a persistence property (i.e. robust). Thus the 
growth rate can be estimated by ignoring the slow drift. The geometric 
singular perturbation theory implies the transient nonlinear 
growth induced by the linear growth rate.

The corresponding linear Navier-Stokes 
operator at ($U(y), 0, W(y)$) is given by the following counterpart of 
(\ref{LC1})-(\ref{LC4}), 
\begin{eqnarray}
& & i(k_1U+k_3W-\sg ) u+ U'v  = - ik_1 p +\e [u'' - (k_1^2+k_3^2)u] , \label{vLC1} \\
& & i(k_1U+k_3W-\sg ) v   = - p' +\e [v'' - (k_1^2+k_3^2)v] , \label{vLC2} \\
& & i(k_1U+k_3W-\sg ) w +W'v  = - ik_3 p +\e [w'' - (k_1^2+k_3^2)w] , \label{vLC3} \\   
& & i k_1 u + v' + ik_3w = 0 . \label{vLC4}
\end{eqnarray}
The simplified system as the counterpart of (\ref{ray1}) is
\begin{equation}
\frac{\hat{\e}}{i\al} \left [ \frac{d^2}{dy^2} - \al^2 \right ]^2 v + H'' v -(H-c) 
\left [ \frac{d^2}{dy^2} - \al^2 \right ]v = 0,
\label{nsi}
\end{equation}
with the boundary condition $v=v'=0$ at $y=a,b$; where $\hat{\e} = \e \sqrt{1+\k^2}$, as before 
$\k = k_3/k_1$ (fixed), $\al = k_1 \sqrt{1+\k^2}$, $c = \sg / k_1$, 
and $H=U+\k W$. Now both equation (\ref{rie}) and  equation (\ref{nsi}) are formally in the same form as 
those in \cite{LL10}. We can specify
\[
a=0, \ b=1, \ \al =0, \ \be = 1 ;
\]
and consider the following sequence of type II 3D shears
\begin{equation}
(U_n(y), 0, W_n(y)) =\left ( y, 0,\frac{A}{n} \sin (4n\pi y) \right ).
\label{sps}
\end{equation}
The problem of linear instability of the type II 3D shears (\ref{sps}) is then casted into the same problem 
as that of the 2D shears (\ref{2s}). By the results of \cite{LL10}, we have
\begin{theorem}
For any $A>0$ and any integer $n \geq 1$, the type II 3D shears (\ref{sps}) are linearly unstable under the 
3D Euler dynamics. Specifically, there exists a 2D unstable eigenmode surface ($k_1,k_3;\sg (k_1,k_3)$)
with $\text{Im}\{ \sg (k_1,k_3)\} >0$ for the equation (\ref{ray1}), stemming from neutral modes of the 
form ($k_1(n),k_3(n);\frac{1}{2}k_1(n)$) where $[k_1(n)]^2 + [k_3(n)]^2 \geq C n^2$ ($C>0$ is independent 
of $n$). The corresponding eigenfunctions are in $C^\infty (0,1)$.
\label{rth}
\end{theorem}
\begin{proof}
Fix $\k = k_3/k_1$ such that
\[
\frac{1}{2}\frac{1}{4\pi} < \k A < \frac{1}{4\pi} ,
\]
then the problem is reduced to that of Theorem 3.2 in \cite{LL10}.
\end{proof}

\begin{theorem}
Let ($k_1^0,k_3^0;\sg^0$) be a point on the 2D unstable eigenmode surface given by Theorem \ref{rth}. Then 
when $\e$ is sufficiently small, there exists an unstable eigenmode ($k_1^0,k_3^0;\sg^*$) with 
$\text{Im}\{ \sg^*\} >0$ for the equation (\ref{nsi}) with $H$ given by (\ref{sps}). When $\e \ra 0^+$,
$\sg^* \ra  \sg^0$.
\end{theorem}
\begin{proof}
For the fixed ($k_1^0,k_3^0$), $\k$ has a fixed value; then $\hat{\e}$ and $\e$ are equivalent. The 
problem is reduced to that of Theorem 4.1 in \cite{LL10}.
\end{proof}
\begin{remark}
Notice that the sequence of type II 3D shears (\ref{sps}) are linearly unstable for all $A>0$, while 
the sequence of 2D shears (\ref{2s}) are proved linearly unstable for 
$\frac{1}{2}\frac{1}{4\pi} < A < \frac{1}{4\pi}$. The sequence of type II 3D shears (\ref{sps}) also
approaches the linear shear $U=y$ in the $L^\infty$ norm of velocity, but not in the 
$L^\infty$ norm of vorticity.
\end{remark}

\end{document}